\documentclass[5p]{elsarticle}
\journal{Computers \& Security} 

\usepackage{lineno}
\modulolinenumbers[5]

\usepackage[utf8]{inputenc}
\usepackage{amsmath,amsthm,amssymb,multirow,tabularx}
\usepackage{parskip}
\PassOptionsToPackage{hyphens}{url}
\usepackage{hyperref}
\usepackage[dvipsnames]{xcolor}
\usepackage{comment}
\usepackage{soul}
\usepackage{caption}
\usepackage{subcaption}
\usepackage{tikz}
\usetikzlibrary{patterns}
\usetikzlibrary{positioning,shapes.geometric,fit,calc}
\usepackage{pgf}
\usepackage{pgfplots}
\pgfplotsset{compat=1.10}
\usepackage{colortbl}
\definecolor{Gray}{gray}{0.95}

\newcommand{\policyholder}{\mathcal{P}}

\newcommand{\insurer}{\mathcal{I}}
\newcommand{\prob}{\Pr}

\newcommand{\calU}{\ensuremath{\mathcal{U}}}

\newcommand{\UP}[2]{\ensuremath{\mathcal{U}}^{P_{\text{#1}}}_{\text{#2}}}
\newcommand{\UI}[2]{\ensuremath{\mathcal{U}}^{\mathcal{I}}_{\mathcal{P}_{\text{#1}}\text{#2}}}

\usepackage[textsize=scriptsize, textwidth=8em]{todonotes}

\newcommand{\Manos}[1]{\todo[color=blue!10]{\textbf{MP}: #1}}

\usepackage[normalem]{ulem}
\newcommand{\change}[2]{\textcolor{red}{\sout{#1}}\textcolor{ForestGreen}{#2}}

\newif\ifHighlightRevision
\HighlightRevisionfalse 
\newenvironment{revision}{\ifHighlightRevision\color{blue}\fi}{\color{black}}

\newtheorem{theorem}{Theorem}

\graphicspath{ {./figs/} }

\usepackage{lineno}

\begin{document}

\begin{frontmatter}


\title{Post-Incident Audits on Cyber Insurance Discounts}


\author[label1]{Sakshyam Panda}
\ead{s.panda@surrey.ac.uk}

\author[label2]{Daniel W Woods}
\ead{daniel.woods@cs.ox.ac.uk}

\author[label3]{Aron Laszka} 
\ead{alaszka@uh.edu}

\author[label4]{Andrew Fielder}
\ead{andrew.fielder@imperial.ac.uk}

\author{Emmanouil Panaousis\corref{cor1}\fnref{label1}}
\ead{e.panaousis@surrey.ac.uk}
\cortext[cor1]{corresponding author}

\address[label1]{University of Surrey}
\address[label2]{University of Oxford}
\address[label3]{University of Houston}
\address[label4]{Imperial College London}

\begin{abstract}
\frenchspacing
We introduce a game-theoretic model to investigate the strategic interaction between a cyber insurance policyholder whose premium depends on her self-reported security level and an insurer with the power to audit the security level upon receiving an indemnity claim. Audits can reveal fraudulent (or simply careless) policyholders not following reported security procedures, in which case the insurer can refuse to indemnify the policyholder. However, the insurer has to bear an audit cost even when the policyholders have followed the prescribed security procedures. As audits can be expensive, a key problem insurers face is to devise an auditing strategy to deter policyholders from misrepresenting their security levels to gain a premium discount. This decision-making problem was motivated by conducting interviews with underwriters and reviewing regulatory filings in the US; we discovered that premiums are determined by security posture, yet this is often self-reported and insurers are concerned by whether security procedures are practised as reported by the policyholders.

To address this problem, we model this interaction as a Bayesian game of incomplete information and devise optimal auditing strategies for the insurers considering the possibility that the policyholder may misrepresent her security level. To the best of our knowledge, this work is the first theoretical consideration of post-incident claims management in cyber security. Our model captures the trade-off between the incentive to exaggerate security posture during the application process and the possibility of punishment for non-compliance with reported security policies. Simulations demonstrate that common sense techniques are not as efficient at providing effective cyber insurance audit decisions as the ones computed using game theory.


\end{abstract}


\end{frontmatter}

\section{Introduction}
\frenchspacing


No amount of investment in security eliminates the risk of loss \cite{anderson2010security}. Driven by the frequency of cyber attacks, risk-averse organizations increasingly transfer residual risk by purchasing cyber insurance. As a result, the cyber-insurance market is predicted to grow to between \$7.5 and \$20 billion by 2020, as identified in~\cite{romanosky2016examining}. 

Similar to other types of insurance, cyber-insurance providers pool the risk from multiple policyholders together and charge a premium to cover the underlying risk.
Yet cyber risks like data breaches are qualitatively different from traditional lines like property insurance. For instance, buildings are built once according to building regulations, whereas computer systems continually change as mobile devices blur the network perimeter and software evolves with additional product features and security patches. Adversaries update strategies to exploit vulnerabilities emerging from technological flux.

Further, the problems of moral hazard and adverse selection become more pressing. Adverse selection results from potential clients being more likely to seek insurance if they face a greater risk of loss. Meanwhile, information asymmetry limits insurers in assessing the applicant's risk. The risk depends on computer systems with many devices in different configurations, users with a range of goals, and idiosyncratic organizational security teams, policies, and employed controls. Collecting information is a costly procedure, let alone assessing and quantifying the corresponding risk.

Moral hazard occurs when insureds engage in riskier behaviour in the knowledge that the insurer will indemnify any losses. Even if initial assessment reveals that security policies are in place, it is no guarantee that they will be followed given that ``a significant number of security breaches result from employees' failure to comply with security policies'' \cite{beautement2009compliance}. Technological compliance suffers too, as evidenced by the Equifax breach resulting from not patching a publicly known vulnerability \cite{moore2017harms}.  

Insurance companies collect risk information about applicants to address adverse selection. We interviewed $9$ underwriters in the UK and found that $8$ of them use self-reported application forms; $7$ of them use telephone calls with the applicant; $3$ of them use external audits; and only one uses on-site audits\footnote{Note that these are mutually inclusive events.}. This suggests that the application process relies on accurate self-reporting of risk factors. Cyber insurance application forms collect information about questions ranging from generic business topics to questions related to information security controls \cite{woods2017mapping}.   

Romanosky et al. \cite{romanosky2017content} introduced a data set resulting from a US law requiring insurers to file documents describing their pricing schemes.  Pricing algorithms depended on the applicant's industry, revenue, past-claims history, and---most relevant to this paper---the security controls employed by the organization. The insurer collects all this information and sets a price according to the formulas described in \cite{romanosky2017content}, reducing the premium when security controls are reported to be in place. This was corroborated by interviews with insurance professionals in Sweden \cite{franke2017cyber}. Surprisingly, individual underwriters determine the size of the discount for security controls on a case-by-case basis, even though this can be as large as 25\% of the premium.

Moral hazard is generally addressed by including terms in the policy that insureds must follow for their coverage to be valid. An early study found that coverage was excluded for a ``failure to take reasonable steps to maintain and upgrade security'' \cite{kesan2005cyberinsurance}. A study from $2017$ found few exclusions prescribing security procedures but the majority of policies contained exclusions for ``dishonest acts'' \cite{romanosky2017content}. One such dishonest act is violating the \emph{principle of up-most good faith} requiring insureds to not intentionally deceive the company offering the insurance \cite{thoyts2010insurance}. 

This principle and the corresponding exclusion mitigates moral hazard, which might otherwise drive honest firms to de-prioritize compliance with security procedures. Further, it imposes a cost on fraudulent organizations claiming that entirely fictional security products are in place to receive a lower premium. For example, one insurer refused to pay out on a cyber policy because ``security patches were no longer even available, much less implemented'' \cite{columbia2016complaint} despite the application form reporting otherwise. We do not consider the legality of this case, but include it as evidence that insurers conduct audits to establish whether there are grounds for refusing coverage.  

Further, insurers offer discounts for insureds based on security posture and often rely on self-reports that security controls are in place. Interviewing insurers revealed concerns about whether security policies were being complied with in reality. Besides, larger premium discounts increase the incentive to misrepresent security levels potentially necessitating a higher frequency of investigation which is uneconomical for insurers. To explore how often should insurers audit cyber insurance claims, we develop a game-theoretic model that takes into account relevant parameters from pricing data collected by analyzing $26$ cyber insurance pricing schemes filed in California and identifies different optimal auditing strategies for insurers. Our analytical approach relies on \textit{Perfect Bayesian Equilibrium} (PBE). We complement our analysis with simulation results with parameter values from the collected data. We further make ``common sense" assumptions regarding auditing strategies and show that in general, insurers are better-off with the game-theoretic strategies. The results will be of interest to policymakers in the United States and the European Union, who believe cyber insurance can improve security levels by offering premium discounts~\cite{woods2017policy}.        


The remainder of this paper is organized as follows. Section \ref{section:relatedwork} identifies existing approaches to modeling the cyber insurance market. We introduce our game-theoretic model in Section \ref{section:model} and present the analysis in Section \ref{section:decisionanalysis}. Section~\ref{section:modelevaluation} details the our methodology for data collection which instantiate our simulation results. Finally, we end with concluding remarks in Section \ref{section:conclusion}.

\section{Related Work} \label{section:relatedwork}
\frenchspacing
This paper continues the trend towards rectifying the ``substantial discrepancy'' \cite{bohme2010modeling} between early cyber insurance models and informal claims about the insurance market. Early research considered factors relevant to the viability of a market. Interdependent security occurs when the risk ``depends on the actions of others'' \cite{kunreuther2003interdependent,laszka2014survey}. Optimists argued that insurers could coordinate the resulting collective action problem \cite{ogut2005,bolot2008new}, leading to a net social welfare gain and a viable market. Skeptics instead focused on the ``high correlation in failure of information systems'' \cite{bohme2006,baer2007cyberinsurance,laszka2014estimating}, citing it as a major impediment to the supply of cyber insurance. Recent empirical work \cite{romanosky2017content} analyzing 180 cyber insurance filings shows that the cyber insurance market is viable.

Beyond viability, researchers explored how insurers can intervene by sharing information, assessing the security of service providers, and investing in software quality. The insurer sharing information about claims data was shown to increase social welfare in \cite{woods2018monte}. Khalili et al. \cite{khalili2018embracing} show that underwriting service providers improves both insurer profit and social welfare. Laszka et al. \cite{laszka2015should} found that the insurer directly investing in software quality can ``reduce non-diversifiable risks and can lead to a more profitable cyber insurance market''.  

The timing of the insurer's intervention plays is an important strategic aspect. Ex-ante interventions for the insurer include risk assessments and security investments before the policy term begins. Shetty et al.~\cite{shetty2010competitive} investigated an insurer who could assess security levels perfectly or not at all, concluding that the latter cannot support a functioning market.~Majuca et al.~\cite{majuca2006evolution} showed that ex-ante assessments in combination with discounts for adopting security controls can lead to an increase in social welfare.~A more recent model introduces stochastic uncertainty about the policyholder's security level~\cite{laszka2018cyber}. 

None of these adverse selection studies consider the potential for insureds to misrepresent their security posture.~Allowing malicious insureds to ``subvert insurer monitoring'' in both the application process and over the policy period was studied in~\cite{schwartz2013cyber}.~The analysis showed that no cyber insurance market could exist.~However, we know from~\cite{columbia2016complaint} that insurers audit insureds and refuse coverage for fraudulent claims.~Our model deviates from ~\cite{schwartz2013cyber} by allowing the insurer to audit claims and withdraw coverage if the insured misrepresents information.

Beyond ex-ante assessment, insurers make decisions regarding ex-post claims management. These decisions have received less attention.~The impact of secondary losses on the policyholder's incentive to claim could lead to over-priced products~\cite{bandyopadhyay2009managers}. Further,~insurers can aggregate claims information to increase social welfare~\cite{woods2018monte}. Empirically it has been suggested insurers will refuse ``claims arising from the insured's failure to maintain security levels'', but the strategic aspects of insurers investigating the incidents leading to claims has not been considered~\cite{kesan2005cyberinsurance}.

The literature on economic theory of insurance fraud has developed two main approaches: \textit{costly state verification} and \textit{costly state falsification} \cite{picard2013economic}. The costly state falsification approach assesses the client's behaviour towards a claim. We consider the costly state verification approach, which focuses on the insurer identifying fraudulent claims. The insurer can verify the claims via auditing but has to bear a verification cost. 
The optimal claim handling usually involves random auditing \cite{picard1996auditing}.

Our contribution to the literature is the first theoretical consideration of post-incident claims management. Our model captures the trade-off between the incentive to exaggerate security posture to receive a premium discount and the possibility of punishment for non-compliance with the reported security policies. We consider misrepresenting security posture a strategic choice for the insured and allow the insurer to respond by auditing claims. Not allowing the insurer to do so leads to market collapse \cite{schwartz2013cyber}. 

\section{Model}  \label{section:model}
\frenchspacing
We model the interaction between the policyholder $\policyholder$ and insurer $\insurer$ as a one-shot dynamic game called the \textit{Cyber Insurance Audit Game} (CIAG), which is represented in Figure \ref{fig:one_shot_insurance_model_nature}. Each decision node of the tree represents a state where the player with the move has to choose an action. The leaf nodes present the payoffs of the players for the sequence of chosen actions. The payoffs are represented in the format $\binom{x}{y}$, where $x$ and $y$ are the payoffs of $\policyholder$ and $\insurer$, respectively. Table \ref{tab:list_symbols} presents the list of symbols used in our model. 
Note that the initial wealth of the policyholder ($W$) and the premium for insurance coverage ($p$) are omitted from the tree for ease of presentation.  

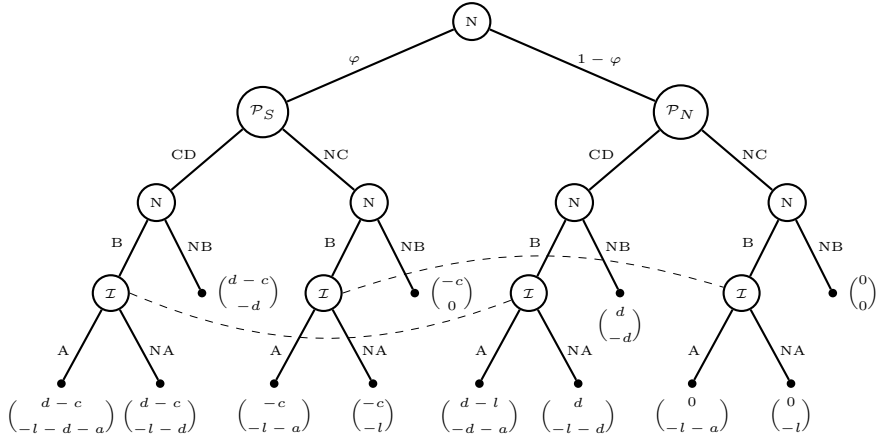
\begin{figure*}[h!] 
\centering
\begin{tikzpicture}[scale=1,thick,-,text centered,level distance=1.2cm,
  font=\tiny,
  solidnode/.style={circle,draw,inner sep=1mm},
  hollownode/.style={circle,draw,inner sep=.8mm,fill=white,line width=2},
  end/.style = {circle,draw,inner sep=0.3mm, fill=black},
  level 1/.style={sibling distance=5.5cm},
  level 2/.style={sibling distance=2.8cm},
  level 3/.style={sibling distance=1.2cm},
  level 4/.style={sibling distance=1.3cm}]
  
\node (root) [solidnode] {N}
  child{ 
    node[solidnode] {$\policyholder_S$}
    child{
      node[solidnode] {N}
        child{
          node(1)[solidnode] {$\insurer$}
            child{
              node[end, label=south:{$\displaystyle\binom{d-c}{-l-d-a}$}] {}
              edge from parent node[left, xshift=-1, yshift=-2] {A}
            }
            child{
              node[end, label=south:{$\displaystyle\binom{d-c}{-l-d}$}] {}
              edge from parent node[right, xshift=0, yshift=-2] {NA}
            }
          edge from parent node[left, xshift=0, yshift=2] {B}  
        }
        child{
          node[end, label=right:{$\displaystyle\binom{d-c}{-d}$}] {}
          edge from parent node[right, xshift=-2, yshift=3]{NB}
        }
      edge from parent node[left, xshift=0, yshift=2]{CD}
    }
    child{
      node[solidnode] {N}
        child{
          node(2)[solidnode] {$\insurer$}
            child{
              node[end, label=south:{$\displaystyle\binom{-c}{-l-a}$}] {}
              edge from parent node[left, xshift=-1, yshift=-2] {A}
            }
            child{
              node[end, label=south:{$\displaystyle\binom{-c}{-l}$}] {}
              edge from parent node[right, xshift=0, yshift=-2] {NA}
            }
          edge from parent node[left, xshift=0, yshift=2] {B}  
        }
        child{
          node[end, label=right:{$\displaystyle\binom{-c}{0}$}] {}
          edge from parent node[right, xshift=-2, yshift=3]{NB}
        }
      edge from parent node[right, xshift=-2, yshift=2]{NC}
    }
    edge from parent node[left, xshift=0, yshift=2]{$\varphi$}
  }
  child{ 
    node[solidnode] {$\policyholder_N$}
      child{
        node[solidnode] {N}
          child{
            node(3)[solidnode] {$\insurer$}
              child{
                node[end, label=south:{$\displaystyle\binom{d-l}{-d-a}$}] {}
                edge from parent node[left, xshift=-1, yshift=-2] {A}
              }
              child{
                node[end, label=south:{$\displaystyle\binom{d}{-l-d}$}] {}
                edge from parent node[right, xshift=0, yshift=-2] {NA}
              }
            edge from parent node[left, xshift=0, yshift=2] {B}
          }
          child{
            node[end, label=below:{$\displaystyle\binom{d}{-d}$}] {}
            edge from parent node[right, xshift=-2, yshift=3] {NB}
          }
        edge from parent node[left, xshift=0, yshift=2]{CD}
      }
      child{
        node[solidnode] {N}
          child{
            node(4)[solidnode] {$\insurer$}
              child{
                node[end, label=south:{$\displaystyle\binom{0}{-l-a}$}] {}
                edge from parent node[left, xshift=-1, yshift=-2] {A}
              }
              child{
                node[end, label=south:{$\displaystyle\binom{0}{-l}$}] {}
                edge from parent node[right, xshift=0, yshift=-2] {NA}
              }
            edge from parent node[left, xshift=0, yshift=2] {B}
          }
          child{
            node[end, label=right:{$\displaystyle\binom{0}{0}$}] {}
            edge from parent node[right, xshift=-2, yshift=3] {NB}
          }
        edge from parent node[right, xshift=-2, yshift=2]{NC}
      }
    edge from parent node[right, xshift=-2, yshift=2]{$1-\varphi$}
  };

\draw[dashed, very thin](1.east) to[bend right=22] (3);
\draw[dashed, very thin](2.east) to[bend left=18] (4);

\end{tikzpicture}
\vspace{-0.3cm}
\caption{Extensive form representation of the Cyber Insurance Audit Game (CIAG) with the Nature deciding the types of policyholder and the occurrence of an incident.}
\label{fig:one_shot_insurance_model_nature}
\end{figure*}

\begin{table}[h!]
\footnotesize
\centering
\caption{List of Symbols}
\label{tab:list_symbols}
\renewcommand*{\arraystretch}{0.75}
\begin{tabular}{|c|l|}
  \hline
  \textbf{Symbol} & \textbf{Description} \\ 
  \hline
  $a$    			        & Cost of audit \\ 
  \rowcolor{Gray}$c$    	& Security investment cost \\
  $d$                       & Discount on premium for better security level \\
  \rowcolor{Gray}$l$  		& Loss due to a breach\\
  $p$ 				        & Premium for the insurance coverage \\
  \rowcolor{Gray}$W$        & Initial wealth of the policyholder \\
  $\beta$    		        & Probability of a breach \\
  \rowcolor{Gray}$\beta^*$  & Probability of a breach after investment \\
  \hline
\end{tabular}
\end{table}

We assume that the policyholder does not make a decision regarding its security investment in our model, because that decision has been made before seeking insurance. Hence, a particular applicant has a certain fixed type (with respect to security), but the insurer does not know the type of an applicant due to information asymmetry. We can model the insurer's uncertainty by assuming that it encounters certain types of applicants with certain probabilities. The type of the policyholder is modeled as an outcome of a random event, that is, nature (N) decides the policyholder's type with respect to additional security investments, i.e., $\policyholder_S$ represents one with additional security investments and $\policyholder_N$ one without. \begin{revision} Further, nature also decides whether a security incident occurs for each policyholder, represented as B (breach) and NB (no breach). \end{revision} The probability of an incident depends on the type of the policyholder. 

Nature moves first by randomly choosing the policyholder's type according to a known \textit{a priori} distribution: $\policyholder_S$ with probability $\prob(\policyholder_S) = \varphi$ and $\policyholder_N$ with probability $\prob(\policyholder_N) = 1 - \varphi$, $\varphi \in [0,1]$. The type is private to a policyholder and the insurer knows only the probability distribution over the different types. Hence, the game is of incomplete information. Regardless of the types, the policyholder's actions are CD (claim premium discount) and NC (no discount claim). Nature then decides the occurrence of the breach on a policyholder, followed by the insurer's decision to audit (A) or not audit (NA) only in the event of a breach. We assume that in CIAG, an audit investigates the misrepresentation of the cyber security investment and the claim for receiving a premium discount. In particular, it investigates whether the policyholder had indeed invested in cyber security countermeasures before claiming this discount. Our model does not assume that there is a particular type of audit.

Having described the players and actions, in the following we present the interaction between $\policyholder$ and $\insurer$.~First, $\policyholder$ has signed up for a cyber insurance contract by paying a premium $p$.~The type of $\policyholder$ is decided by the nature based on an additional security investment. We assume that this investment equals $c$. This investment will decrease the probability of $\policyholder$ being compromised from $\beta$ to $\beta^*$. 

At the same time the investment will enable $\policyholder$ to claim a premium discount $d$. We assume that $\insurer$ offers $d$ without performing any audit since investigating at this point would mean that $\insurer$ would have to audit policyholders who might never file an indemnity claim, thereby incurring avoidable losses.  

We further assume that if $\policyholder$ decides to claim a discount without making the security investment, she will still receive $d$ but risks having a future claim denied after an audit. After an incident, where $\policyholder$ suffers loss $l$, insurer $\insurer$ has to decide whether to conduct an audit (e.g. forensics) to investigate details of the incident including the security level of $\policyholder$ at the time of breach. We assume that this audit costs $a$ to the insurer. This audit will result in: 
    \begin{itemize}
        \item[Case 1:] confirming that $\policyholder$ has indeed invested in security as claimed, in which case $\insurer$ will pay the indemnity. We assume full coverage so the indemnity payment equals $l$.
        
        \item[Case 2:] discovering that $\policyholder$ has misrepresented her security level, $\insurer$ refuses to pay the indemnity and $\policyholder$ has to bear the incident cost $l$. We assume that this case falls within the contract period during which $\policyholder$ is locked-in by the contract. We define \textit{misrepresentation} as when $\policyholder$ is fraudulent or simply careless in maintaining the prescribed security level in the insurance contract and reports a fabricated security level to get the premium discount.
    \end{itemize}

In Figure \ref{fig:one_shot_insurance_model_nature}, some decision nodes of $\insurer$ are connected through dotted lines indicating that the $\insurer$ cannot distinguish between the connected nodes due to unknown $\policyholder$ type. These sets of decision nodes define the \textit{information sets} of the insurer. An information set is a set of one or more decision nodes of a player that determines the possible subsequent moves of the player conditioning on what the player has observed so far in the game. The insurer also has two information sets, one where the breach has occurred to the policyholder who has claimed premium discount CD$=\{$(CD$|\policyholder_S$), (CD$|\policyholder_N$)$\}$ and the one where the breach has occurred to the policyholder who has not claimed premium discount NC$=\{$(NC$|\policyholder_S$), (NC$|\policyholder_N$)$\}$. Each of the insurer's information sets has two separate nodes since the insurer does not know the real type of the policyholder when deciding on whether to audit or not.     

In outcome CD,B,A, the expected utility of the $\policyholder_S$ is 
\begin{equation}
    \UP{S}{CD,B,A} = U(W - p + d - c) , \nonumber
\end{equation}
where $U$ is a utility function, which we assume to be monotonically increasing and concave, $W$ is the policyholder's initial wealth, $p$ is the premium paid to the insurer, $d$ is the premium discount, and $c$ is the cost of the security investment. We assume the utility function to be concave to model the risk aversion of policyholders as defined in \cite{bohme2010modeling}. Note that we assume that $W > p > d$ and $W - p + d > c$, and both $W$ and $p$ are exogenous to our model.
\begin{align}
    \UP{S}{CD,B,A} = \UP{S}{CD,B,NA} = \UP{S}{CD,NB} &= U(W - p + d - c) \\
    \UP{S}{NC,B,A} = \UP{S}{NC,B,NA} = \UP{S}{NC,NB} &= U(W - p - c) \\
    \UP{N}{CD,B,A} &= U(W - p + d - l) \\
    \UP{N}{CD,B,NA} = \UP{N}{CD,NB} &= U(W - p + d) \\
    \UP{N}{NC,B,A} = \UP{N}{NC,B,NA} = \UP{N}{NC,NB} &= U(W - p) 
\end{align}

We further assume that the policyholder's goal is to maximize her expected utility. The expected utility of the policyholder is influenced by the possibility of a breach and the insurer's probability to audit. In particular, the expected utility for $\policyholder_S$ will be the same regardless of the insurer's probability to audit and the breach probability due to indemnification. $\policyholder_N$, however, will need to consider these probabilities. 

In the outcome $\policyholder_S,$CD$,$B$,$A the insurer's utility is
\begin{equation}
    \UI{S}{,CD,B,A} = p - l - d - a , \nonumber
\end{equation}
where $p$ is the premium, $d$ is the premium discount offered, $l$ is the loss claimed by the policyholder, and $a$ is the audit cost.

In other outcomes, the insurer's utility is as follows:
\begin{align}
    \UI{S}{,CD,B,A} &= p - l - d - a \\
    \UI{S}{,CD,B,NA} = \UI{N}{,CD,B,NA} &= p - l - d \\
    \UI{S}{,CD,NB} = \UI{N}{,CD,NB} &= p - d \\
    \UI{S}{,NC,B,A} = \UI{N}{,NC,B,A} &= p - l - a \\
    \UI{S}{,NC,B,NA} = \UI{N}{,NC,B,NA} &= p - l \\
    \UI{S}{,NC,NB} = \UI{N}{,NC,NB} &= p \\
    \UI{N}{,CD,B,A} &= p - d - a
\end{align}

\section{Decision Analysis} \label{section:decisionanalysis}
\frenchspacing
In this section, we analyze the equilibria of the proposed Cyber Insurance Audit Game (Figure \ref{fig:one_shot_insurance_model_nature}), which is a dynamic Bayesian game with incomplete information.~The analysis is conducted using the game-theoretic concept of \textit{Perfect Bayesian Equilibrium} (PBE).~This provides insights into the strategic behaviour of the policyholder $\policyholder$ concerning discount claims and the insurer~$\insurer$'s auditing decision. 

A PBE, in the context of our game, can be defined by Bayes requirements discussed in \cite{gibbons1992primer}:
\begin{itemize}
    \item[] \textbf{Requirement 1:} The player at the time of play must have a belief about which node of the information set has been reached in the game. The beliefs must be calculated using Bayes' rule, whenever possible, ensuring that they are consistent throughout the analysis. 
    
    \item[] \textbf{Requirement 2:} Given these beliefs, a player's strategy must be sequentially rational. A strategy profile is said to be sequentially rational if and only if the action taken by the player with the move is optimal against the strategies played by all other opponents given the player's belief at that information set.
    
    \item[] \textbf{Requirement 3:} The player must update her beliefs at the PBE to remove any implausible equilibria. These beliefs are determined by Bayes' rule and players' equilibrium strategies. 
\end{itemize}

In the event of a security breach, the insurer's decision to audit or not must be based on beliefs regarding the policyholder's types.~More specifically, a belief is defined as a probability distribution over the nodes within the insurer's information set, conditioned that the information node has been reached. The insurer has two information sets subjected to whether the policyholder has claimed premium discount or not which are CD$=\{$(CD$|\policyholder_S$), (CD$|\policyholder_N$)$\}$ and NC$=\{$(NC$|\policyholder_S$), (NC$|\policyholder_N$)$\}$. The insurer assigns a belief to each of these information sets. Let $\mu$ and $\lambda$ be the insurer's beliefs where
\begin{align}
    \mu = \prob(\policyholder_S|\text{CD}) \nonumber \\
    \lambda = \prob(\policyholder_S|\text{NC}) \nonumber
\end{align}
That is, for the first information set, the insurer believes with $\mu$ and $1-\mu$ that the premium discount claim is from $\policyholder_S$ and $\policyholder_N$, respectively. Similarly, for the second information set, the insurer believes with $\lambda$ that $\policyholder_S$ has not claimed premium discount and believes with $1-\lambda$ that $\policyholder_N$ has not claimed premium discount.  

The first requirement of PBE dictates that Bayes' rule should be used to determine beliefs. Thus
\begin{align}
    \mu = \frac{\prob(\policyholder_S)\prob(\text{CD}|\policyholder_S)}{\prob(\policyholder_S)\prob(\text{CD}|\policyholder_S) + \prob(\policyholder_N)\prob(\text{CD}|\policyholder_N)} \label{eq:mu}\\
    \lambda = \frac{\prob(\policyholder_S)\prob(\text{NC}|\policyholder_S)}{\prob(\policyholder_S)\prob(\text{NC}|\policyholder_S) + \prob(\policyholder_N)\prob(\text{NC}|\policyholder_N)} \label{eq:lambda}
\end{align}

From the payoffs in Figure \ref{fig:one_shot_insurance_model_nature}, it can be clearly seen that CD is always a preferred choice for $\policyholder_S$. Whereas, the insurer always gets a better payoff for choosing NA against NC irrespective of policyholder's type.~Having defined the necessary concepts, next, we identify the possible PBEs of the game for the following constraints 
\begin{equation} \label{constraint_1}
    l > a \text{ and } l > d 
\end{equation}
\begin{equation} \label{constraint_2}
    l > a \text{ and } l < d 
\end{equation}
\begin{equation} \label{constraint_3}
    l < a \text{ and } l > d 
\end{equation}
\begin{equation} \label{constraint_4}
    l < a \text{ and } l < d
\end{equation}
where the PBEs are strategy profiles and beliefs that satisfies all the three requirements described earlier. 

\vspace{\baselineskip}
\begin{theorem}
For $\varphi > \frac{l-a}{l}$, $l>a$ and $l>d$, CIAG has only one pure-strategy PBE \big(\textnormal{(CD,CD)},\textnormal{(NA,NA)}\big), in which the policyholder claims premium discount regardless of her type while the insurer does not audit regardless of whether the policyholder claims or not a discount, with $\mu = \varphi$ and arbitrary $\lambda \in [0, 1]$.
\end{theorem}

\begin{proof} The existence of pure-strategy PBE can be verified by examining the strategy profile (CD,CD) and (NA,NA) with constraint in Equation~\eqref{constraint_1}. 
This represents the case where an incident has occurred on the policyholders who have claimed premium discount.    
    \begin{itemize}
        \item[a)] \textit{Belief consistency}: Due to information asymmetry and as only one of the insurer's information set is in the equilibrium path, she assigns $\prob(\text{CD}|\policyholder_S) = 1$ and $\prob(\text{CD}|\policyholder_N) = 1$. Thus, using Bayes' rule in Equation~\eqref{eq:mu} gives
        \begin{align}
            \mu = \varphi/(\varphi + 1 - \varphi) = \varphi \nonumber
        \end{align}
        On the other hand, applying Bayes' rule in Equation \eqref{eq:lambda} to $\lambda$ yields $0/0$ which is an indeterminate result. This implies that if the equilibrium is actually played then the off-equilibrium information set NC should not be reached restricting an update to the insurer's belief with Bayes' rule. Due to an indeterminate result, the insurer specifies an arbitrary $\lambda \in [0,1]$.
        
        \item[b)] \textit{Insurer's sequentially rational condition given updated beliefs}: The expected payoff for each action of the insurer are
        
        \begin{align}
            \calU_\text{A} &= \varphi \cdot \UI{S}{,CD,B,A} + (1 - \varphi) \cdot \UI{N}{,CD,B,A} \\
            &= \varphi(p-d-l-a) + (1-\varphi)(p-d-a) \nonumber \\
            &= p - \varphi l - d - a \nonumber
        \end{align}
        \begin{align}
            \calU_{\text{NA}} &= \varphi \cdot \UI{S}{,CD,B,NA} + (1 - \varphi) \cdot \UI{N}{,CD,B,NA} \\
            &= \varphi(p-d-l) + (1-\varphi)(p-d-l) \nonumber \\
            &= p - d - l \nonumber
        \end{align}
        The condition for A to be sequentially rational is $\calU_\text{A} > \calU_{\text{NA}}$ which gives 
        \begin{align}
            p - \varphi l - d -a &> p - d - l \nonumber \\
            \varphi & \le \frac{l-a}{l} = \varphi^*
        \end{align}
        Now considering the off-equilibrium information set NC, the insurer always gets a better payoff by choosing NA. Thus, NA is a dominant strategy of the insurer against the off-equilibrium information set NC. The insurer's belief $\lambda$ remains arbitrary. 
        
        \item[c)] \textit{Policyholder's sequentially rational condition given insurer's best response}: Knowing the best responses of the insurer i.e. (A,NA) for $\varphi \le \varphi^*$ and (NA,NA) for $\varphi > \varphi^*$ against CD, we derive the best response of the policyholder.
        For insurer's strategy profile (NA,NA), $\policyholder_S$ gets a payoff $U(W - p + d - c)$ by choosing CD. If she deviates to NC, she will get a payoff $U(W - p - c)$ which is undesirable. Whereas, $\policyholder_N$ receives a payoff $U(W - p + d)$ by choosing CD. If she deviates to NC will get a payoff $U(W - p)$ which is also undesirable. Thus, (CD,CD) and (NA,NA) can be verified as a PBE given $\varphi > \varphi^*$ and $\mu = \varphi$. Note that the PBE includes the updated beliefs of the insurer implicitly satisfying Requirement 3.     
    \end{itemize}
\end{proof}    

From the PBE, we can see that if $l>a$, $l>d$ and insurer's belief $\varphi$ is greater than the threshold value $\varphi^*$, not auditing a breach is optimal for the insurer and claiming premium discount is optimal for the policyholder regardless of her type. When the insurer's belief $\varphi \le \varphi^*$ there exist no pure-strategy PBE. As a result, both players will mix up their strategies. We discuss this mixed-strategy PBE below. Note that, in the following, we use the inner tuple ($x, 1-x$) to indicate a mixed strategy where the player chooses the first action with probability $x$ and the second action with probability $1-x$.      

\vspace{\baselineskip}
\begin{theorem}
For $\varphi \le \frac{l-a}{l}$, $l>a$, $l>d$, CIAG has only one mixed-strategy PBE, in which:

\begin{itemize}
\item $\policyholder_S$ will always prefer CD, while $\policyholder_N$ randomizes between CD and NC with probability $\delta$ and $1-\delta$, respectively;

\item the insurer randomizes between A and NA with probability $\theta$ and $1-\theta$, respectively, against CD, and she always prefer NA against NC, with her beliefs about $\policyholder_S$ playing CD and NC being $\overline{\mu} = \frac{\varphi}{\varphi + (1-\varphi)\delta} \quad$ and $\quad \overline{\lambda} = 0 \label{eq:mudash_lambdadash}$, respectively, where 
\begin{revision}
\begin{align}
    \delta &= \frac{a}{(1-\varphi)l} \nonumber \\
    \theta &= \frac{U(W - p + d) - U(W - p)}{\beta \cdot \Big(U(W - p + d) - U(W - p + d - l) \Big)} \label{eq:mixedprob}
\end{align}
\end{revision}
\end{itemize}
\end{theorem}

\begin{proof}
    The existence of mixed-strategy PBE is outlined below.
    \begin{itemize}
        \item[a)] \textit{Belief consistency}: Again we apply the Bayes' rule. By assuming that the policyholder sticks to the equilibrium strategy, the insurer can derive that $\prob(\text{CD}|\policyholder_S)$ $= 1$, $\prob(\text{NC}|\policyholder_S) = 0$, $\prob(\policyholder_S) = \varphi$, $\prob(\policyholder_N) = 1-\varphi$, $\prob(\text{CD}|\policyholder_N) = \delta$ and $\prob(\text{NC}|\policyholder_N) = 1-\delta$. Using Equations \eqref{eq:mu} and \eqref{eq:lambda} we obtain $\mu = \overline{\mu}$ and $\lambda = \overline{\lambda}$ in Equation \eqref{eq:mudash_lambdadash}.
    
        \item[b)] \textit{Optimal responses given beliefs and opponent's strategy}: Given these beliefs and the mixed strategy of the policyholder, an insurer's optimal strategy would maximize her payoff. The insurer can achieve this by randomizing her actions such that the expected payoffs is equal for all the actions of the policyholder. This is known as the \textit{indifference principle} in game theory. Thus, the expected utility of $\policyholder_N$ for choosing CD is
        
        \begin{align}
            \UP{N}{CD} =& 
            \beta \cdot \Big( \theta \cdot \UP{N}{CD,B,A} + (1 - \theta) \cdot \UP{N}{CD,B,NA} \Big) \nonumber \\
            & + (1 - \beta) \cdot \UP{N}{CD,NB} \nonumber \\ 
            = & \beta \cdot \theta \cdot U(W - p + d - l) \nonumber \\
            & + \beta \cdot (1 - \theta) \cdot U(W - p + d) \nonumber \\
            & + (1 - \beta) \cdot U(W - p + d) \nonumber \\
            = & \beta \cdot \theta \cdot \Big(U(W - p + d - l) - U(W-p+d) \Big) \nonumber \\
            &+ U(W - p + d) 
        \end{align}
        and for choosing NC, where NA is a dominating strategy of the insurer,~is
        \begin{align}
            \UP{N}{NC} &= \beta \cdot \UP{N}{NC,B,NA} + (1-\beta) \cdot \UP{N}{NC,NB} \nonumber \\
            &= \beta \cdot U(W - p) + (1 - \beta) \cdot U(W - p) \nonumber \\
            &= U(W - p) 
        \end{align}
        The indifference principle requires that $\calU^{\policyholder_N}_{\text{NC}} = \calU^{\policyholder_N}_{\text{CD}}$, which gives 
        \begin{align}
            U(W - p) &= \beta \! \cdot \! \theta \! \cdot \Big(\!U(W - p + d - l) - U(W - p + d) \! \Big) \nonumber \\
            &\quad+ U(W - p + d)  \nonumber \\
            \theta &= \frac{U(W - p + d) - U(W - p)}{\beta \cdot \Big(U(W - p + d) - U(W - p + d - l) \Big)} \nonumber
        \end{align}
        as in Equation \eqref{eq:mixedprob}. Similarly, the policyholder will also mix her strategy with an aim to make the insurer indifferent between choosing A and NA. Thus, 
        \begin{align}
            \calU^\insurer_\text{A} &= \varphi \cdot \UI{S}{,CD,B,A} \nonumber \\
            &\quad + (1-\varphi) \cdot \Big( \UI{N}{CD,B,A} + \UI{N}{NC,B,A} \Big) \nonumber \\
            &= \varphi \Big((1)(p-d-l-a) + (0)(p-l-a) \Big) \nonumber \\ 
            &\quad + (1-\varphi) \Big(\delta(p-d-a) + (1-\delta)(p-l-a) \Big) \nonumber \\
            &= p - l -a - \varphi d - \delta d + \delta l + \varphi \delta d - \varphi \delta l 
        \end{align}
        \begin{align}
            \calU_\text{NA} &= \varphi \cdot \UI{S}{,CD,B,NA} \nonumber \\
            &\quad + (1-\varphi) \cdot \Big( \UI{N}{CD,B,NA} + \UI{N}{NC,B,NA} \Big) \nonumber \\
            &= \varphi \Big((1)(p-d-l) + (0)(p-l) \Big) \nonumber \\ 
            &\quad + (1-\varphi) \Big(\delta(p-d-l) + (1-\delta)(p-l) \Big) \nonumber \\
            &= p - l - \varphi d - \delta d + \varphi \delta d 
        \end{align}
        and $\calU_{\text{A}} = \calU_{\text{NA}}$ gives 
        \begin{align}
            p - l -a - \varphi d - \delta d + \delta l + \varphi \delta d - \varphi \delta l =& p - l - \varphi d \nonumber \\
            & - \delta d + \varphi \delta d \nonumber \\
            \delta =& \frac{a}{(1-\varphi)l} \nonumber
        \end{align}
    \end{itemize}
     as in Equation \eqref{eq:mixedprob}.
\end{proof}



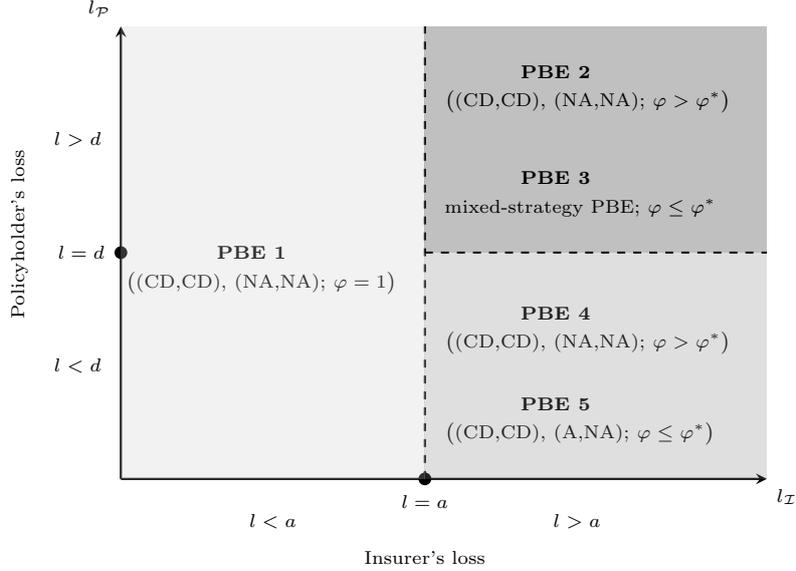
\begin{figure*}[h] 
\centering
\begin{tikzpicture}[thick,>=stealth,font=\scriptsize,
  dot/.style = {draw,fill = black,circle,inner sep = 0pt,minimum size = 4pt}
  ]
  \draw[thick,->] (0,0) -- (8.5,0) node[anchor=north west] {$l_{\insurer}$};
  \draw[thick,->] (0,0) -- (0,6) node[anchor=south east] {$l_{\policyholder}$};
  
  \draw (4,0) node[dot, label= {below:$l=a$}] {};
  \draw (2,-0.2) node[label={below:$l<a$}] {};
  \draw (6,-0.2) node[label={below:$l>a$}] {};
  \draw (4,-0.7) node[label={below:Insurer's loss}] {};
  \draw[dashed,-] (4,0) -- (4,6) {};
  \draw (0,3) node[dot, label= {left:$l=d$}] {};
  \draw (0,1.5) node[label={left:$l<d$}] {};
  \draw (0,4.5) node[label={left:$l>d$}] {};
  \draw (-1.2,4.5) node[label={[rotate=90]left:Policyholder's loss}] {};
  \draw[dashed,-] (4,3) -- (8.5,3) {};
  
  \draw (1,3) node[label={right:\textbf{PBE 1}}] {};
  \draw (-0.2,2.6) node[label={right:$\big($(CD,CD), (NA,NA); $\varphi=1 \big)$}] {};
  \draw (5,5.4) node[label={right:\textbf{PBE 2}}] {};
  \draw (4,5) node[label={right:$\big($(CD,CD), (NA,NA); $\varphi > \varphi^* \big)$}] {};
  \draw (5,4) node[label={right:\textbf{PBE 3}}] {};
  \draw (4,3.6) node[label={right:mixed-strategy PBE; $\varphi \le \varphi^*$ }] {};
  \draw (5,2.2) node[label={right:\textbf{PBE 4}}] {};
  \draw (4,1.8) node[label={right:$\big($(CD,CD), (NA,NA); $\varphi > \varphi^* \big)$}] {};
  \draw (5,1) node[label={right:\textbf{PBE 5}}] {};
  \draw (4,0.6) node[label={right:$\big($(CD,CD), (A,NA); $\varphi \le \varphi^* \big)$}] {};
  
  \fill[black!20,nearly transparent] (0,0) -- (4,0) -- (4,6) -- (0,6);
  \fill[black!50,nearly transparent] (4,0) -- (8.5,0) -- (8.5,3) -- (4,3) -- cycle;
  \fill[black!99,nearly transparent] (4,3) -- (8.5,3) -- (8.5,6) -- (4,6) -- cycle;
\end{tikzpicture}
\caption{Solution space of Cyber Insurance Audit Game (CIAG).}
\label{fig:one_shot_solution_space_graph}
\end{figure*}

We conceive all the possible PBEs for CIAG by exhaustively applying this methodology over all combinations of the players' strategy profiles for the four constraints described in Equations \eqref{constraint_1} to \eqref{constraint_4}. Figure \ref{fig:one_shot_solution_space_graph} presents the solution space of CIAG. It further shows how the equilibrium strategies of the players depends on the premium discount (d), audit cost (a), and loss~(l).

\section{Model Evaluation} \label{section:modelevaluation}
Our analysis in Section \ref{section:decisionanalysis} provides a framework for insurers to determine optimal auditing strategy against policyholders who can misrepresent their security levels to avail premium discounts. This section illustrates the methodology used to obtain values for various parameters of our model and simulation results using these values to determine the best strategy for the insurer.   

\subsection{Methodology and Data Collection}
\frenchspacing
A diverse set of data sources is needed to study the interaction between insurance pricing, the effectiveness of security controls, and the cost of auditing claims. To this end, we combine the following data sources: a US law requiring insurers to report pricing algorithms \cite{romanosky2017content}, analysis of a data set of over $12,000$ cyber events \cite{romanosky2016examining}, a study of the cost and effectiveness of security controls \cite{heitzenrater2016policy}, and a range of informal estimates regarding the cost of an information security audit.

The model assumes that nature determines incidents according to a Bernoulli distribution with loss amount $l$ and probability of loss $\beta$.  
Analysis of the data set of $12,000$ cyber incidents reveals data breach incidents occur with a median loss \$$170$K and frequency of around $0.015$ for information firms \cite{romanosky2016examining}, which we use as $l$ and $\beta$ respectively. 

We adopt the security control model used in \cite{heitzenrater2016policy}. Both fixed and operational costs are estimated using industry reports, which correspond to $c$ in our model. The effectiveness of a control is represented as a percentage decrease in the size or frequency of losses. For example, operating a firewall (\$2,960) is said to reduce losses by $80$\% \cite{heitzenrater2016policy}---leading to a probability of breach after investment ($\beta$*) of $0.2 \beta$. 

We downloaded all of the cyber insurance filings in the state of California and discarded off-the-shelf policies that do not change the price based on revenue, industry or security controls. This left $26$ different pricing algorithms and corresponding rate tables, the contents of which are described in \cite{romanosky2017content}.

Data breach coverage with a \$$1$ million limit was selected because it is the default coverage and it comfortably covers the loss value $l$ for a data breach on SMEs. The premium $p$ and discount $d$ varies based on the insurer. We chose a filing explicitly mentioning discounts for firewalls. For an information firm with $\$40$M of revenue, the premium $p$ is equal to \$3,630 and the filings provide a range of discounts up to $25$\%. The exact value depends on an underwriter's subjective judgment. To comprise this we consider multiple discounts in this range.

Estimating the insurer's cost of audit ($a$) is difficult because they could be conducted by loss adjusters within the firm or contracted out to IT specialists. With the latter in mind, we explored the cost of an information security audit. The cost depends on the depth of the assessment and the expertise of the assessor. However, collating the quoted figures suggests a range from $\$5,000$ up to $\$100,000$.

\subsection{Numerical Analysis}
We simulate the interaction between the cyber insurance policyholder and the insurer based on our game-theoretic model with parameter values described above. First, we compare the expected payoffs of the insurer for different strategic models: 
\begin{enumerate}
    \item the game-theoretic approach (GT) where the insurer chooses an appropriate strategy according to our analysis (refer to Figure \ref{fig:one_shot_solution_space_graph}) and can either audit or not audit;
    \item always auditing (A,A) regardless of whether the policyholder has claimed discount or not;
    \item always not auditing (NA,NA) regardless of whether the policyholder has claimed discount or not;
    \item auditing if the policyholder has claimed discount and not auditing if there is no discount claimed (A,NA);
    \item not auditing if the policyholder has claimed discount and auditing if there is no discount claimed (NA,A);
    \item auditing half the times regardless of whether the policyholder has claimed discount or not ($0.5$A,$0.5$A);
    \item auditing half the times when the policyholder has claimed discount and not auditing if there is no discount claimed ($0.5$A,NA).
\end{enumerate}

In the following simulation figures, the insurer's average payoffs with each strategic model are calculated against a policyholder who plays the PBE strategy obtained through our analysis. This policyholder is also the most challenging one for the insurer as it claims for a discount even in the case of non investment. The term ``$x$ repetitions of the game" reflects that CIAG is played $x$ number of independent runs for a set of parameter values.

From Figures \ref{fig:ins_avg_with_audit_5k_d_5} and \ref{fig:ins_avg_with_audit_5k_d_25}
we observe that the payoff of the insurer when choosing the GT model is always better than rest of the strategic models irrespective of the premium discount. The reason for this is that the model (A,NA), where the insurer audits only policyholders who have claimed the discount, is susceptible to auditing clients who have implemented additional security level bearing the auditing cost as a pure loss. Thus, the larger the number of honest policyholders, the higher the insurer's loss is. Additionally, the insurer's loss as expected increases with the increasing cost of audit. With the (NA,NA) model, the insurer chooses to reimburse the loss without confirming the policyholder's actual security level. Here, the insurer indemnifies even for cases where the policyholder has misrepresented her security level suffering heavy losses. Another non strategic approach would be to randomize over the choice of auditing or not auditing a policyholder who has claimed a premium discount. This strategy, represented by the model ($0.5$A,NA), gives a payoff within the range of payoffs from models (A,NA) and (NA,NA). The results exhibit that models (A,A), (NA,A), and (0.5A,0.5A) consistently performers poorly compared to other models.  

\begin{figure}[h]
     \centering
     \begin{subfigure}[b]{0.5\textwidth}
         \centering
         \includegraphics[width=\textwidth]{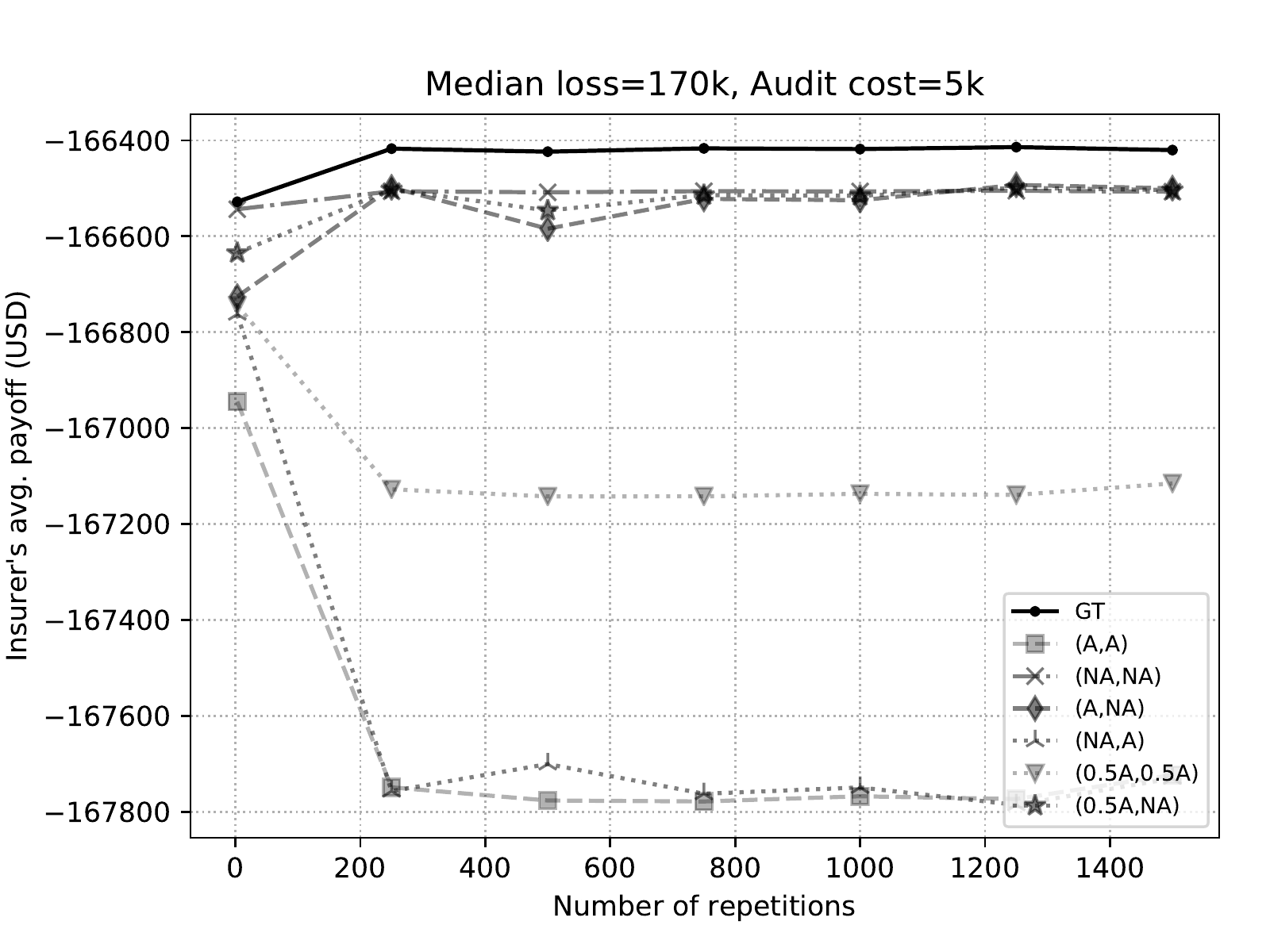}
         \caption{Premium discount $5\%$}
         \label{fig:ins_avg_with_audit_5k_d_5}
     \end{subfigure}
     \hfill
     \begin{subfigure}[b]{0.49\textwidth}
         \centering
         \includegraphics[width=\textwidth]{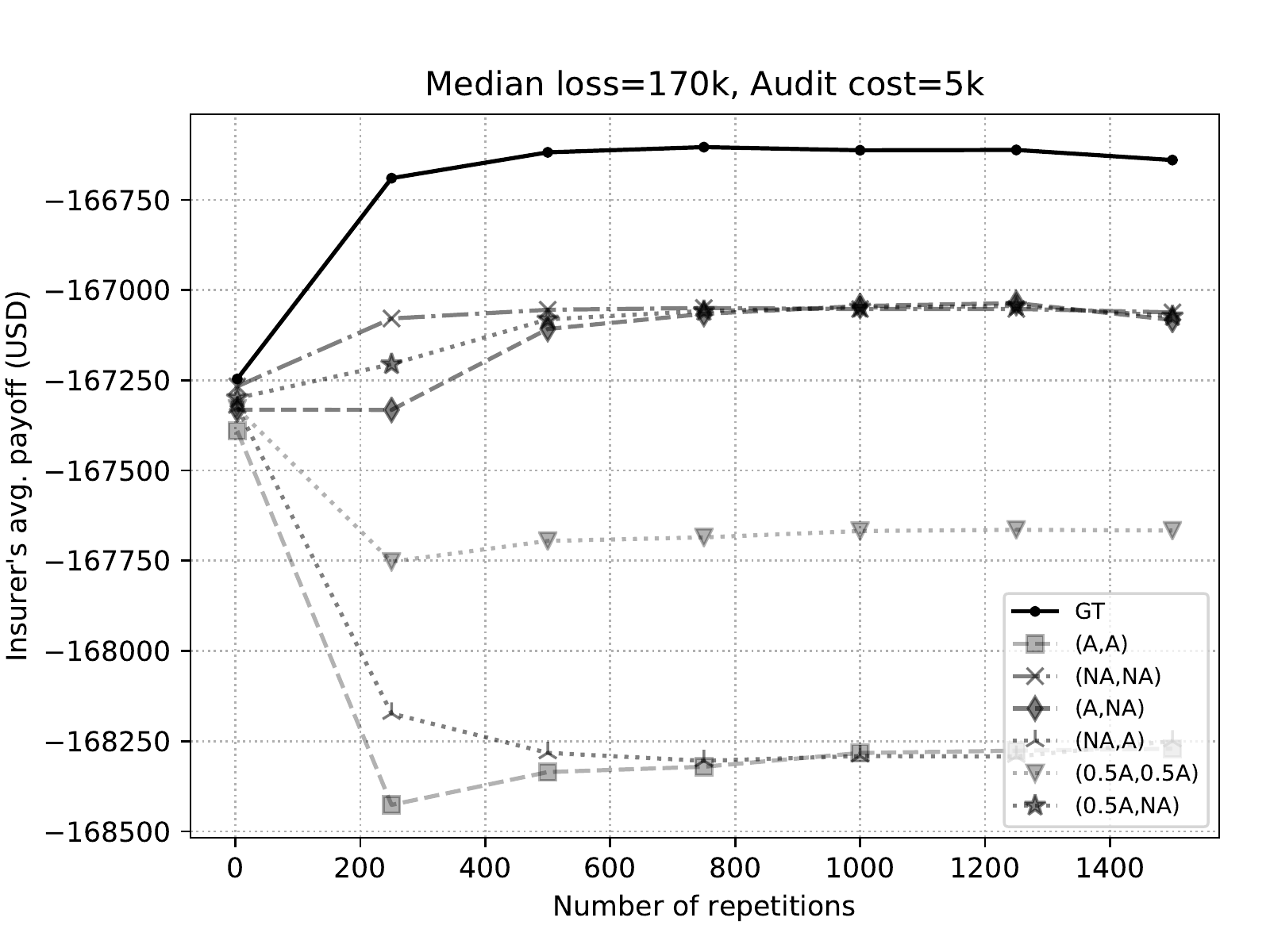}
         \caption{Premium discount $25\%$}
         \label{fig:ins_avg_with_audit_5k_d_25}
     \end{subfigure}
     \caption{Insurer's average payoff (in US dollars) with different strategic models across various repetitions of the game for a median loss $\$170$k, audit cost $\$5$k with (a) premium discount $5\%$ and (b) premium discount $25\%$.}
     \label{fig:ins_avg_a_5k}
\end{figure}

The GT model presents an optimal mix of (A,NA) and (NA,NA) where the insurer's decision to audit is based on \textit{a prior} belief regarding the policyholder's security investment. For the median loss of $\$170$k which is greater than both the audit cost and premium discount, the game solution is derived from the upper-right section of the solution space in Figure \ref{fig:one_shot_solution_space_graph}. In particular, when the insurer's belief ($\phi$) regarding the policyholder's security investment is greater than a threshold ($\varphi^*$), she prefers (NA,NA) i.e, \textbf{PBE 2}: $\big($(CD,CD),(NA,NA); $\varphi>\varphi^* \big)$. When the belief is lower than $\varphi^*$, she prefers a mixed approach (\textbf{PBE 3}) by simultaneously relying on (A,NA) and (NA,NA) and choosing whichever is more profitable. The GT model, thus, enables the insurer to take into account a prior belief regarding the policyholder's security investment under the condition of information asymmetry and maximize her payoffs given this belief. The figures further show that regardless of how many times the game has been played model GT performs better than the non-game-theoretic models. 

\begin{figure}[h]
     \centering
     \begin{subfigure}[b]{0.5\textwidth}
         \centering
         \includegraphics[width=\textwidth]{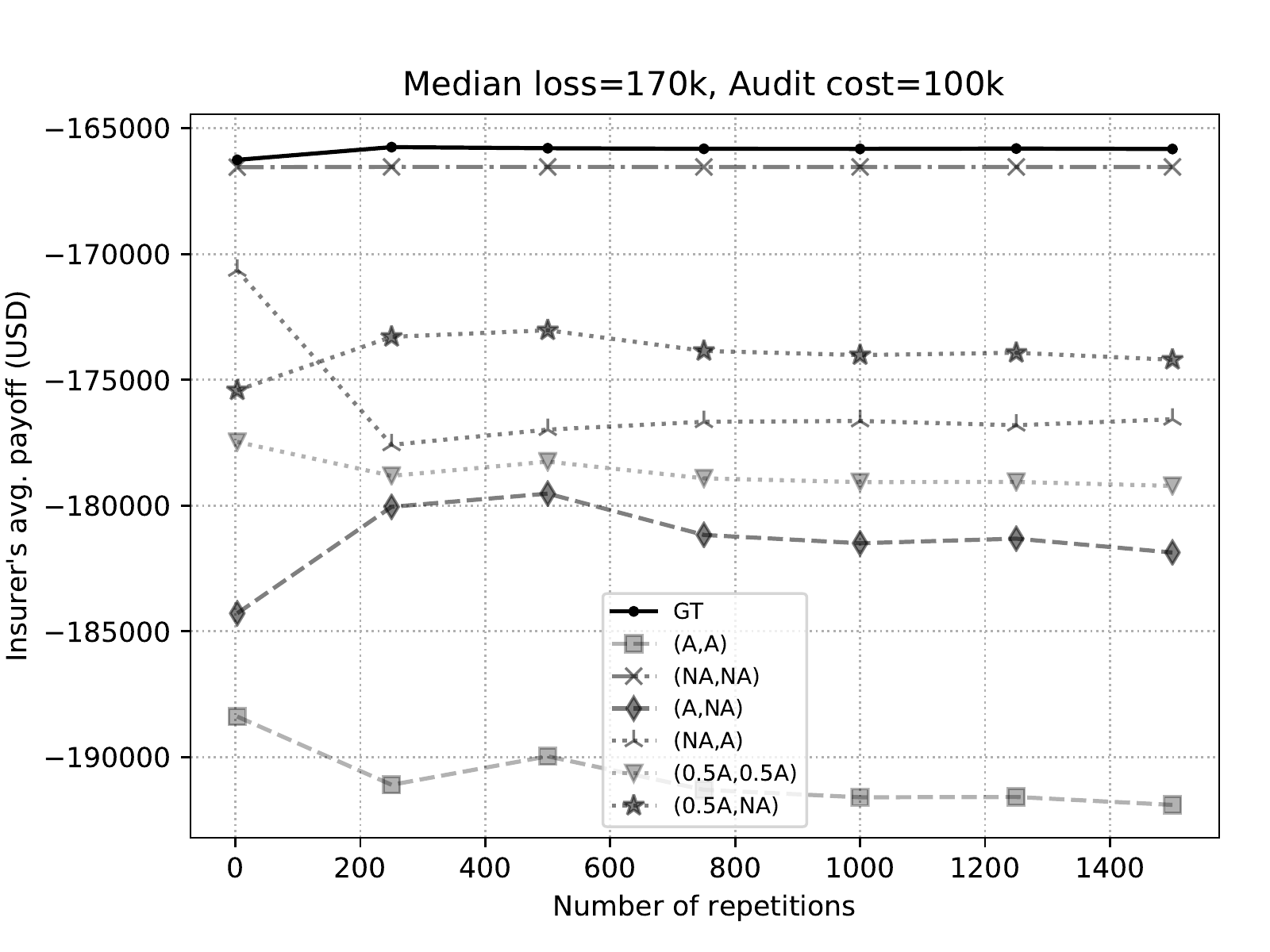}
         \caption{Premium discount $5\%$}
         \label{fig:ins_avg_with_audit_100k_d_5}
     \end{subfigure}
     \hfill
     \begin{subfigure}[b]{0.49\textwidth}
         \centering
         \includegraphics[width=\textwidth]{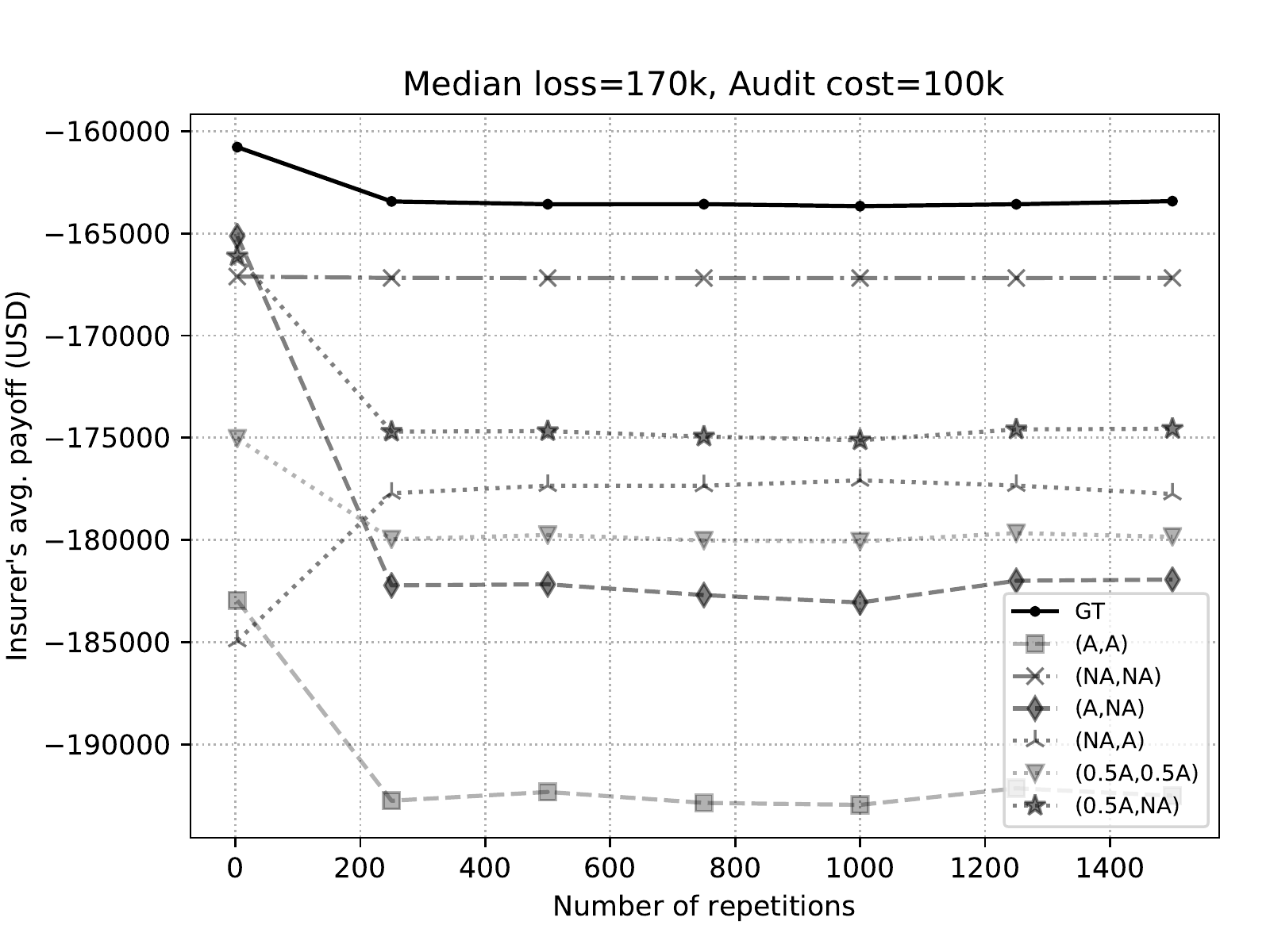}
         \caption{Premium discount $25\%$}
         \label{fig:ins_avg_with_audit_100k_d_25}
     \end{subfigure}
     \caption{Insurer's average payoff (in US dollars) with different strategic models across various repetitions of the game for a median loss $\$170$k, audit cost $\$100$k with (a) premium discount $5\%$ and (b) premium discount $25\%$.}
     \label{fig:ins_avg_a_100k}
\end{figure}

With higher audit cost i.e., $\$100$k in Figures \ref{fig:ins_avg_with_audit_100k_d_5} and \ref{fig:ins_avg_with_audit_100k_d_25}, we observe that the insurer's average payoff with the model (A,NA) decreases drastically confirming it's shortcomings as discussed above. In the case of 1500 independent repetitions for the highest values of audit cost and premium discount, the insurer gains, on average, a higher payoff when choosing GT as opposed to (NA,NA) model. The increased difference in the payoff is equivalent to $98\%$ of the annual premium charged to policyholder.

\begin{figure*}[h!]
     \centering
     \begin{subfigure}[]{0.5\textwidth}
         \centering
         \includegraphics[width=\textwidth]{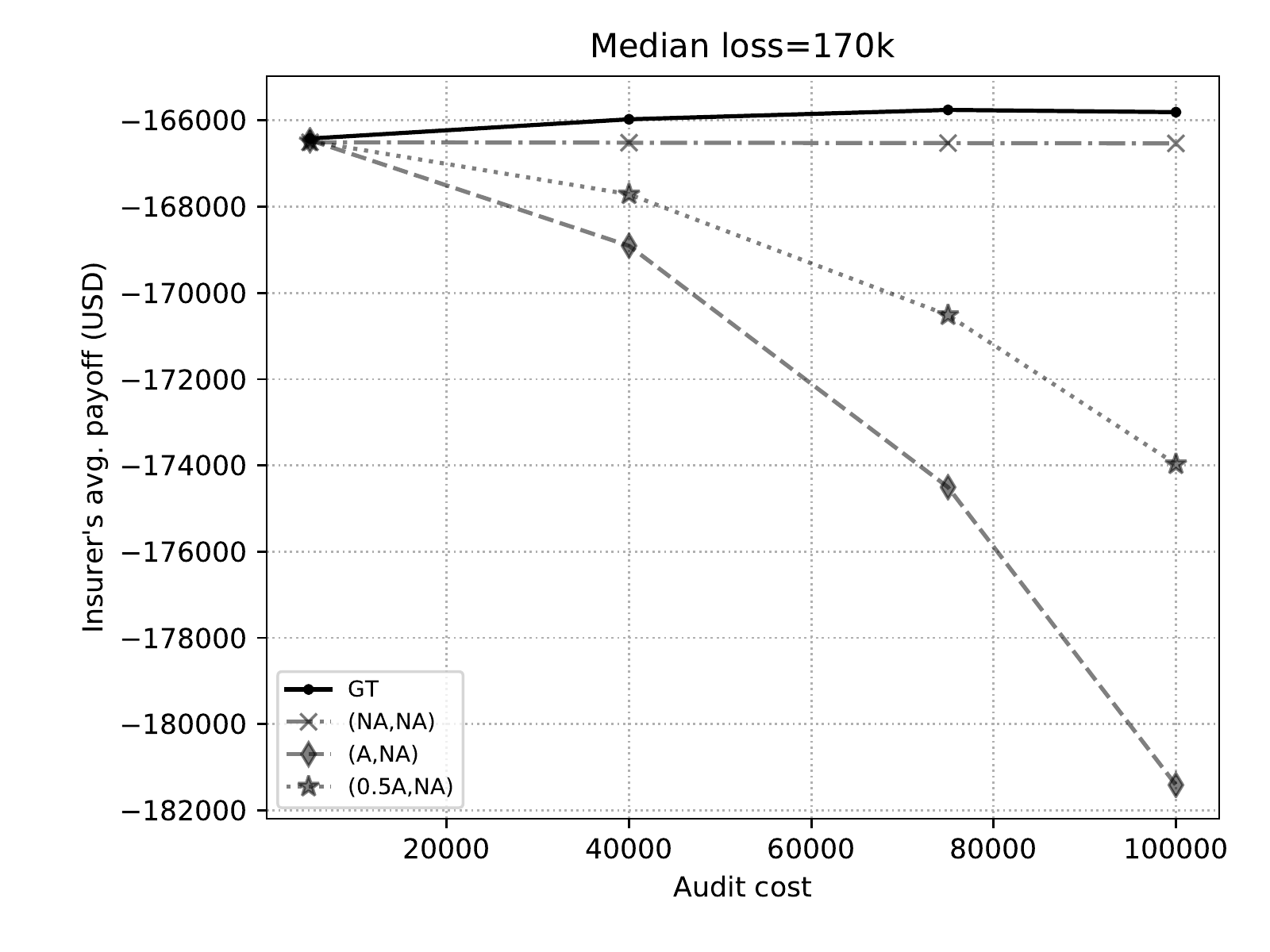}
         \caption{Premium discount $5\%$}
         \label{fig:ins_avg_vs_a_d_5}
     \end{subfigure}
     \hfill
     \begin{subfigure}[]{0.49\textwidth}
         \centering
         \includegraphics[width=\textwidth]{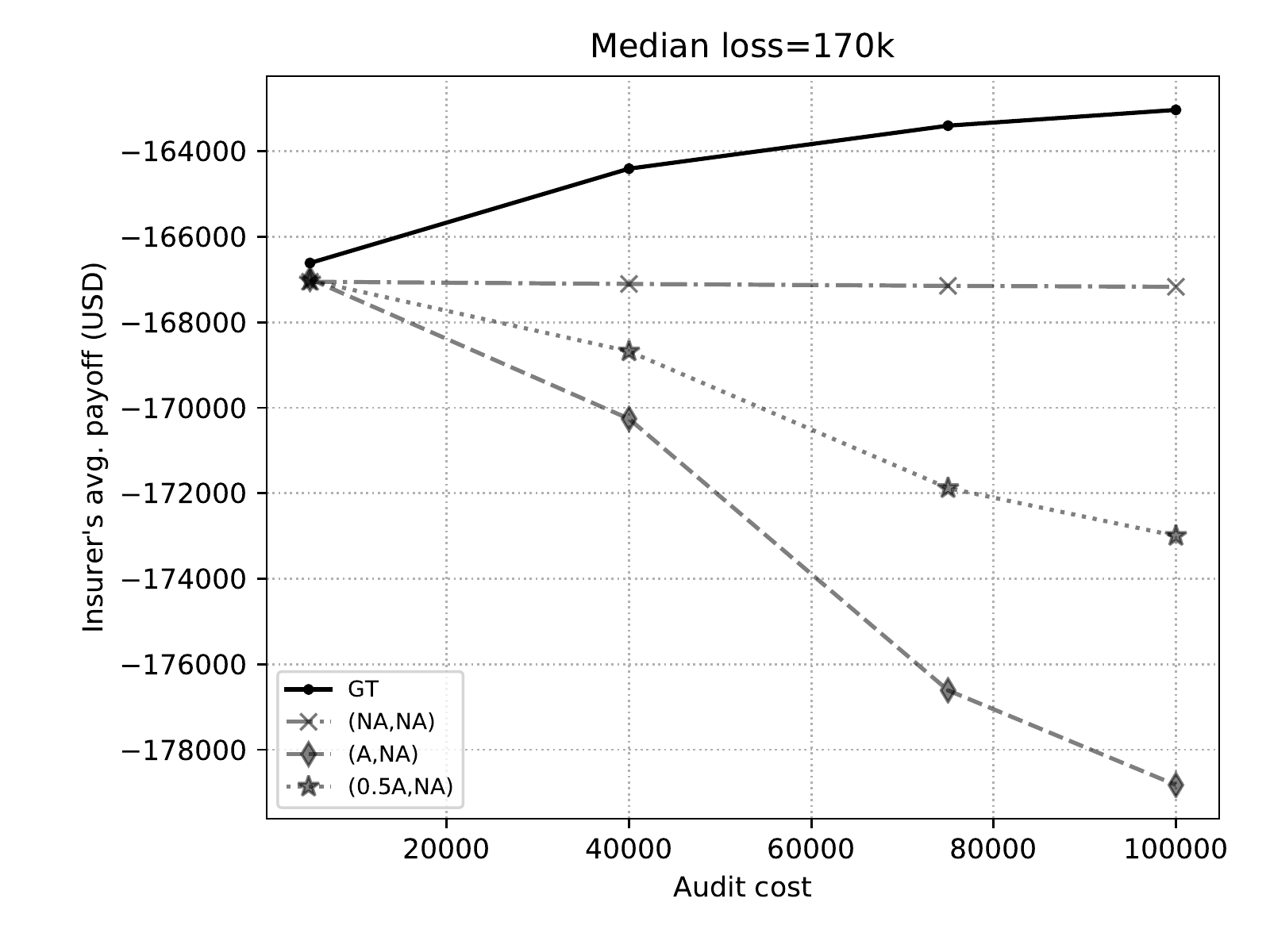}
         \caption{Premium discount $25\%$}
         \label{fig:ins_avg_vs_a_d_25}
     \end{subfigure}
     \hfill
     \begin{subfigure}[]{0.5\textwidth}
         \centering
         \includegraphics[width=\textwidth]{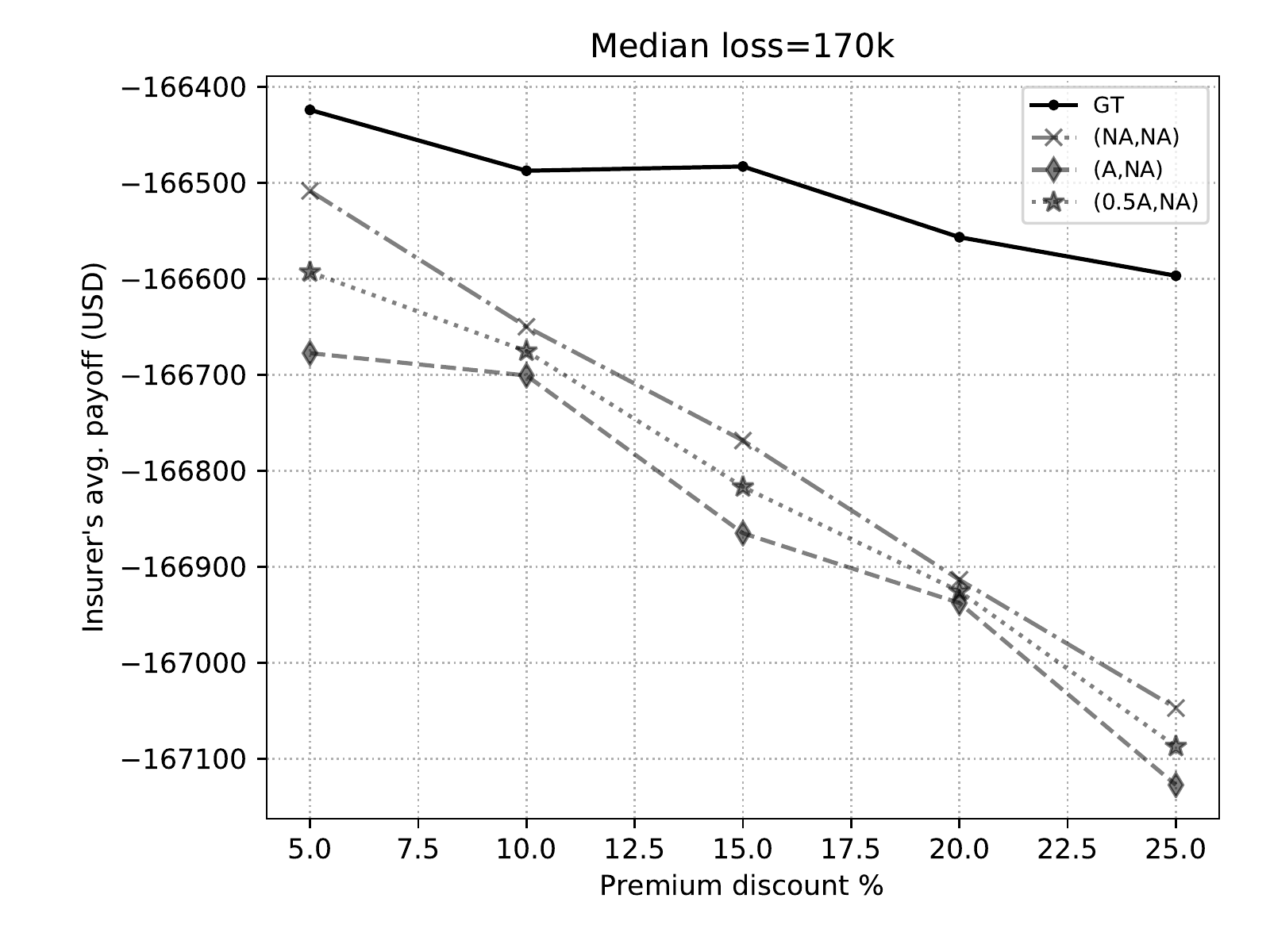}
         \caption{Audit cost $5$k}
         \label{fig:ins_avg_vs_d_a_5}
     \end{subfigure}
     \hfill
     \begin{subfigure}[]{0.49\textwidth}
         \centering
         \includegraphics[width=\textwidth]{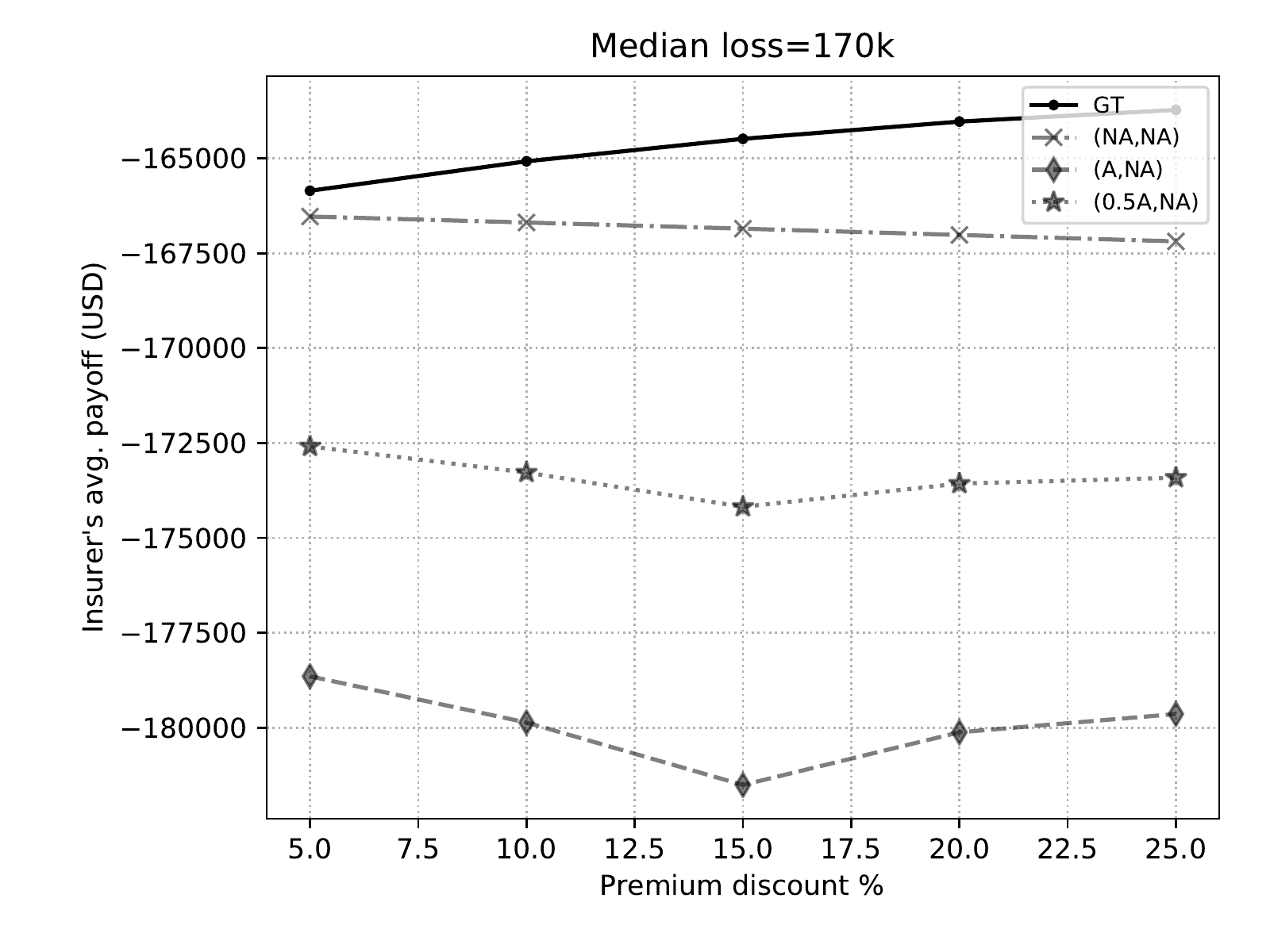}
         \caption{Audit cost $100$k}
         \label{fig:ins_avg_vs_d_a_100}
     \end{subfigure}
     \caption{Insurer's average payoff (in US dollars) with different strategic models for median loss $\$170$k against audit costs with premium discount (a) $5\%$ and (b) $25\%$ and against premium discount with audit cost (c) $5$k and (d) $25$k.}
     \label{fig:ins_avg_vs_ad}
\end{figure*}

\textbf{Remark 1:} For constant loss, as premium discount increases, GT consistently outperforms all other strategic models 
for various repetitions of the game in both sets of experiments with minimum and maximum values of audit cost.

Next, the simulation results are obtained over 100 repetitions with a median loss of $\$170k$ against a range of audit cost, premium discount and loss. Note that the models (A,A), (NA,A), and (0.5A,0.5A) are omitted from the figures as they perform worse than others, and for ease of presentation. 

Figures \ref{fig:ins_avg_vs_a_d_5} and \ref{fig:ins_avg_vs_a_d_25} show that there is a point of convergence where the strategy largely doesn't matter, but then as the audit cost increases, there is motivation for playing the game-theoretic solution as any other solution is worse. As discount increases, a policyholder might be highly stimulated to receive premium discount given that the insurer will grant this without auditing her before an incident occurs. This escalates the possibilities of the policyholder misrepresenting her actual security level. Given this possibility, GT noticeably dominates other strategic models as seen in Figures \ref{fig:ins_avg_vs_d_a_5} and \ref{fig:ins_avg_vs_d_a_100}. Further, in the case of 100 independent repetitions with the highest values of premium discount and audit cost, deploying GT gives the insurer on average a higher payoff compared to the next best model which is (NA,NA). The increased difference in the payoff is equivalent to $60\%$ of the annual premium charged to the policyholder. 

\textbf{Remark 2:} For a constant loss, as premium discount and audit cost increase, GT outperforms all other strategic models.

\begin{figure*}[h]
    \centering
    \includegraphics[width=0.7\textwidth]{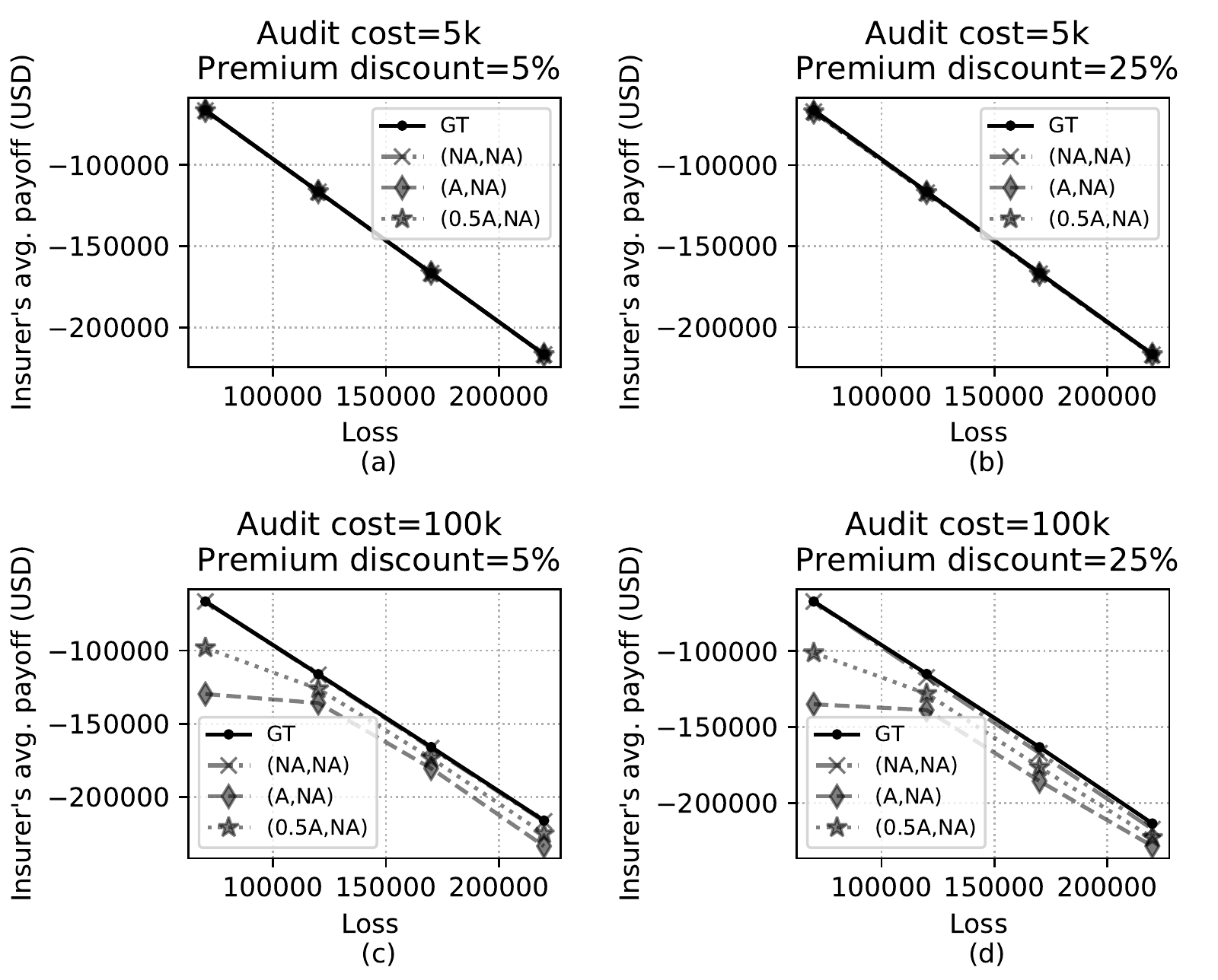}
    \caption{Insurer's average payoff (in US dollars) with different strategic models against loss for various audit costs and premium discounts.}
    \label{fig:ins_avg_vs_l}
\end{figure*}

Figure \ref{fig:ins_avg_vs_l} shows that there is essentially nothing special about the loss as a contributing factor with low audit cost, but become discriminatory as the audit cost approaches the loss. GT and (NA,NA) performs equally well until this condition, but as the discount increases with the audit cost, GT exceeds (NA,NA). In this case, for 100 independent repetitions, insurers gain on average a higher payoff with model GT, compared to model (NA,NA), the next best. The increased difference is payoff is equivalent to $66\%$ of the annual premium charged to the policyholder.

\textbf{Remark 3:} As premium discount, audit cost, and loss increase, GT consistently outperforms all other strategic models.

In summary, we have demonstrated how an insurer may use our framework in practice to determine the best auditing strategy against a policyholder. We have illustrated how the insurer's payoff is maximized by strategically choosing to audit or not in the event of a breach. Such strategic behaviour also allows the insurer to maximize her payoff against policyholders who can misrepresent their security levels to avail premium discount.

\section{Conclusion} \label{section:conclusion}
\frenchspacing
Speaking to cyber insurance providers reveals concerns about the discrepancy between the security policies applicants report that they follow, in the application process, and the applicant's compliance with these policies once coverage is in place. To address this, we developed a game-theoretic framework investigating audits as a mechanism to disincentivize misrepresentation of security level by policyholders. Thus far, we know of one instance \cite{columbia2016complaint} denying cyber insurance coverage due to non-compliance with the security practices as defined in the insurance contract. Although there could have been denials settled in private, this suggests that most cyber insurance providers follow the \emph{never audit} strategy. Our analysis derived a game-theoretic strategy that outperforms na\"ive strategies, such as never audit. By considering the post-incident claims management process, we demonstrated how a cyber insurance market can avoid collapse (contradicting \cite{schwartz2013cyber}) when the policyholder can fraudulently report their security level.

\begin{revision}
To extend this paper, future work could consider modelling uncertainty about the effectiveness of the implemented security measure. In the current model, the policyholder's type is chosen by Nature according to some probability distribution. It could be extended such that the policyholder maximizes expected payoff by selecting an investment strategy based on the beliefs about her type. This consideration would extend, for example, our analysis to consider the overall utility function of the policyholder, that is considering both the investment and no investment types simultaneously, and maximizing the expected payoff. 

Another interesting direction is investigating how the potential loss $l$ changes as a function of the security investment. In this case, we will be looking into different types of risk profiles of the policyholders. We could also investigate the trade-off between the additional investment, discount, and residual risk.
\end{revision}

Finally, a future extension could make investment in security a strategic choice for the policyholder in a multi-round game with a \emph{no claims bonus}, as our data set describes the size of these discounts. We could also allow belief updates to influence insurer choices on each iteration.

\section*{Acknowledgements}
\begin{revision}
We thank the reviewers for their valuable feedback and comments.
\end{revision}

Sakshyam Panda and Emmanouil Panaousis have been partially funded by the European Union’s Horizon 2020 research and innovation programme under the Marie Skłodowska-Curie SECONDO grant agreement No 823997.  

Daniel Woods could participate thanks to a Fulbright Cyber Security Award.

\end{document}